\RequirePackage{amsmath}
\documentclass[runningheads]{llncs}
\usepackage{graphicx}
\usepackage[dvipsnames]{xcolor}
\usepackage[colorlinks=true, allcolors=blue,linkcolor=blue, citecolor=blue,backref=page]{hyperref}
\usepackage{listings}
\usepackage{amssymb}
\usepackage{amsfonts}
\usepackage{mathtools}
\usepackage{maple}
\usepackage[utf8]{inputenc}
\usepackage{breqn}
\usepackage{tabularray}
\usepackage{orcidlink}
\usepackage{mathrsfs}
\usepackage{listings}
\usepackage{microtype}
\usepackage{moreverb}
\usepackage{verbatim}
\usepackage{sverb}
\usepackage{geometry}
\usepackage{makecell}
\usepackage{url}
\usepackage{cleveref}
\usepackage{textcomp}
\usepackage{moreverb}
\usepackage{verbatim}
\usepackage{sverb}
\usepackage{algorithm}
\usepackage{algpseudocode}
\usepackage{algorithmicx}
\usepackage{euscript}

\newcommand*{\QEDA}{\hfill\ensuremath{\blacksquare}}

\newcommand*{\CQFD}{\hfill\ensuremath{\square}}

\newcommand{\mmod}[1]{\ \mathrm{mod}\ #1}
\newcommand{\mfoldind}[3][n\equiv]{\chi_{\{#1 #2\mmod #3\}}}

\newcommand{\NN}{\mathbb{N}}

\newcommand{\ZZ}{\mathbb{Z}}
\newcommand{\QQ}{\mathbb{Q}}
\newcommand{\KK}{\mathbb{K}}

\DeclareMathOperator{\lcm}{lcm}

\begin{document}
    \title{Hypergeometric-Type Sequences}

\author{{\it In memory of Marko Petkov\v{s}ek
}
\\[0.5cm]
Bertrand Teguia Tabuguia\orcidlink{0000-0001-9199-7077}}

\authorrunning{B. Teguia Tabuguia}

\institute{Department of Computer Science, University of Oxford, UK\\
Max Planck Institute for Software Systems, Saarbr\"{u}cken, Germany\\
\email{bertrand.teguia@cs.ox.ac.uk}\\
}

\maketitle              

\begin{abstract}
We introduce hypergeometric-type sequences. They are linear combinations of interlaced hypergeometric sequences (of arbitrary interlacements). We prove that they form a subring of the ring of holonomic sequences. An interesting family of sequences in this class are those defined by trigonometric functions with linear arguments in the index and $\pi$, such as Chebyshev polynomials, $\left(\sin^2\left(n\,\pi/4\right)\cdot\cos\left(n\,\pi/6\right)\right)_n$, and compositions like $\left(\sin\left(\cos(n\pi/3)\pi\right)\right)_n$.

We describe an algorithm that computes a hypergeometric-type normal form of a given holonomic $n\text{th}$ term whenever it exists. Our implementation enables us to generate several identities for terms defined via trigonometric functions.

\end{abstract}

\section{Introduction}\label{sec:intro}

The connection between summation and linear difference equations with polynomial coefficients dates back to Fasenmyer \cite{fasenmyerphd,fasenmyer1949note,WolfBook}. These equations, called holonomic or P-recursive (or P-finite), easily lead to closed forms of the corresponding sums when the equations are made of two non-zero terms. The resulting equations have the form
\begin{equation}
    P(n)\,a_{n+m} = Q(n)\,a_n,
\end{equation}
where $P$ and $Q$ are polynomials. When $m=1$, the corresponding solution is a hypergeometric term. For $m>1$, we say that the solution is $m$-fold hypergeometric.

Zeilberger \cite{zeilberger1990holonomic} and Chyzak \cite{chyzak2000extension} generalized the algorithmic approach to finding holonomic recurrence equations for sums. Zeilberger proposed an efficient algorithm for dealing with sums whose summands are hypergeometric terms \cite{zeilberger1990fast}. It resulted in several automatic proofs of special functions and combinatorial identities in synergy with his collaboration with Wilf \cite{wilf1990towards} for the WZ method, which cleverly uses Gosper's algorithm \cite{gosper1978decision}. The recipe for systematic finding of hypergeometric identities could not be ready before the availability of an algorithm for finding all hypergeometric term solutions of P-recursive equations. For instance, for a summation problem like 
$$
s_n\coloneqq\sum_{k=-\infty}^{\infty} (-1)^k \binom{n}{k} \binom{d\cdot k}{n}, \,\quad\, d\in\NN,
$$
Zeilberger's original algorithm finds a recurrence equation of order $d$ while it can be proven that $s_n=(-d)^n$ -- a hypergeometric term.

In 1993, Petkov\v{s}ek completed the recipe by proposing algorithm \textit{Hyper} \cite{petkovvsek1992hypergeometric}. Thanks to that, the computer algebra community benefited from the book $A=B$ \cite{AeqB}, which gathers all fundamental results, at least for that moment. The present paper does not intend to investigate symbolic summation. For further results in this direction, we recommend the following non-exhaustive list of references \cite{paule1995greatest,koepf1995algorithms,abramov2001minimal,koutschan2013creative,chyzak2009non,chyzakabc,chen2018reduction}, \cite[Chapter 5]{KauersDfinite}.

Petkov\v{s}ek's algorithm did not only serve combinatorial identities but also formal power series. Koepf establishes a connection with his work \cite{koepf1992power}, where he also highlighted the need for an algorithm for finding $m$-fold hypergeometric term solutions of holonomic recurrences. Petkov\v{s}ek and Salvy \cite{petkovvsek1993finding} proposed a way to adapt Hyper to this context. Abramov further investigated the study of such solutions \cite{abramovmsparse}. Ryabenko gave the first concrete implementation in the Maple computer algebra system (CAS) \cite{Ryabenkoformal}.

Mark van Hoeij improved Petkov\v{s}ek's algorithm to a much more efficient version \cite{van1999finite,cluzeau2006computing}. A theoretical algorithm for its generalization to the $m$-fold case was first proposed in \cite{horn2012m}. We here recall results from the author's Ph.D. thesis \cite{BTphd} (see also \cite{BTWKsymbconv}) as a refreshment before introducing hypergeometric-type sequences.

\begin{definition}[Proper hypergeometric-type power series \cite{Teguia2021hypergeometric}]\label{def:hyptypseries} For an expansion around zero, a series $S(z)$ is said to be of ``proper'' hypergeometric type if it can be written as
\begin{equation}\label{eq:defhyptypseries}
    S(z) \coloneqq \sum_{j=1}^J S_{j}(z),\,\quad\, S_{j}(z)\coloneqq \sum_{i=0}^{m_j-1}\sum_{n=0}^{\infty} a_i(m_j\,n+i)\,z^{m_j\,n+i},
\end{equation}
where $m_j,J\in \NN, m_j\neq 0$, and $a_i(n)$ is a linear combination of $m_j$-fold hypergeometric terms.
Thus, a proper hypergeometric-type power series is a linear combination of formal power series whose coefficients are $m$-fold hypergeometric terms. A proper hypergeometric-type function is a function that can be expanded as a proper hypergeometric-type power series.
\end{definition}
The word ``proper'' in \Cref{def:hyptypseries} is used to lighten the definition in \cite{BTWKmaple,Teguia2021hypergeometric} by neglecting Laurent-Puiseux series. Note that in contrast with the definition given in the original papers \cite{BTphd,BTWKsymbconv}, here we highlight that the coefficients $a_i$'s are not necessarily built from the same hypergeometric terms, and this is in perfect agreement with the scope of the formal power series algorithm proposed there.

Recently, Koepf and the author designed \textit{mfoldHyper}, an algorithm that extends the algorithms by Petkov\v{s}ek and van Hoeij to find all $m$-fold hypergeometric term solutions of P-recursive equations. It has the advantage of offering a better efficiency than the algorithm from \cite{petkovvsek1993finding,Ryabenkoformal} (see also \cite{BTvariant}). Algorithm mfoldHyper helped to design a complete algorithm to convert a univariate holonomic function into a hypergeometric-type power series \cite{BTphd,BTWKsymbconv}. The resulting algorithm is available from Maple 2022 as \texttt{convert/FormalPowerSeries}, and mfoldHyper as \texttt{LREtools:-mhypergeomsols}, all from the \texttt{FPS} package at \cite{FPS}. The Maxima version of the package is in the process of being integrated into Maxima.

We observed that the concept of hypergeometric type can be adapted to sequences. In this regard, a similar development as that of formal power series would enable a compact definition of closed forms for terms that are usually expressed in cases depending on some properties satisfied by the index. Roughly speaking, a hypergeometric-type sequence is a sequence whose general term ($n$th term) is a linear combination of $m$-fold hypergeometric terms. We will give a formal definition in the next section. Several examples can be generated with trigonometric functions. In this case, the property satisfied by the index refers to its remainder with respect to some non-negative integer.

Throughout this paper, we will always assume that $n$ in an integer, usually non-negative, i.e., $n\in\NN\coloneqq \{0,1,2,\ldots\}$. It may sound intriguing to notice that there seems to be no CAS that uses a symbolic computation algorithm to find normal forms free of unevaluated trigonometric functions. Some examples are $\sin^2\left(n\,\pi/4\right)$, $\sin\left(\cos\left(n\,\pi/3\right)\,\pi\right)$, etc. Of course, one could always eliminate trigonometric functions using simplifications with Euler's formulas (see, for instance, \texttt{convert/exp} in Maple or \texttt{TrigToExp} in Mathematica); however, this would not define normal forms as the simplifications used after conversion into exponentials may be tailored to the given trigonometric expressions. Because Maple seems to be the only CAS implementing mfoldHyper or its analog from \cite{Ryabenkoformal}, and such a simplification of trigonometric sequences is not available in Maple, it is reasonable to see our idea as a novel method. It sets a symbolic approach for finding closed forms of a more general class of sequences, for which we also develop a theoretical framework.

Our approach to hypergeometric-type sequences resembles that of hypergeometric-type power series. We consider three main steps: given a term $h_n$,

\begin{enumerate}
    \item find a P-recursive equation satisfied by $h_n$;
    \item find a basis of all $m$-fold hypergeometric term solutions of that equation (using mfoldHyper, for instance);
    \item use initial values from $h_n$ to deduce a hypergeometric-type normal form for $h_n$.
\end{enumerate}
The three steps would be successful if $h_n$ is the term of a hypergeometric-type sequence (or simply a hypergeometric-type term).

In the next section, we define hypergeometric-type sequences, study their structure, and state some properties, like their link to hypergeometric-type functions. After discussing canonical and normal forms of hypergeometric-type terms, \Cref{sec:sec2} details the three steps of our algorithmic method. In \Cref{sec:sec3}, we present our current Maple implementation from the package \texttt{HyperTypeSeq}. We recommend using Maple versions between 2019 and 2021 because of some misbehavior of the package with the recent releases. The package is accessible via Github at \cite{HyperTypeSeq}. Below are two simple formulas automatically computed using our implementation. 
\begin{lstlisting}
> with(HyperTypeSeq):
> HTS(sin(n*Pi/4)^2,n)
\end{lstlisting}
\begin{dmath}\label{maple1}
\frac{1}{2}-\frac{\left(-1\right)^{\frac{n}{2}} \chi_{\left\{\mathit{modp} \left(n ,2\right)=0\right\}}}{2}
\end{dmath}

The sequence \href{https://oeis.org/A212579}{A212579} from the OEIS \cite{sloane2003line} satisfies the recurrence equation
\begin{lstlisting}
> RE:= a(n) = a(n-1)+2*a(n-2)-a(n-3)-2*a(n-4)-a(n-5)+2*a(n-6)+a(n-7)-a(n-8)
\end{lstlisting}
\begin{dmath*}
    \mathit{RE}\coloneqq a\! \left(n\right)=a\! \left(n-1\right)+2 a\! \left(n-2\right)-a\! \left(n-3\right)-2 a\! \left(n-4\right)-a\! \left(n-5\right)+2 a\! \left(n-6\right)+a\! \left(n-7\right)-a\! \left(n-8\right).
\end{dmath*}
Using the first thirteen initial values, our algorithm finds the closed form:
\begin{lstlisting}
> REtoHTS(RE,a(n),[0, 1, 8, 31, 80, 171, 308, 509, 780, 1137, 1584, 2143, 2812])
\end{lstlisting}
\begin{dmath}\label{maple2}
\frac{4}{9}+\frac{31}{12} n-3 n^{2}+\frac{67}{36} n^{3}-\frac{1}{4} n \chi_{\left\{\mathit{modp} \left(n,2\right)=0\right\}}-\frac{4}{9} \chi_{\left\{\mathit{modp} \left(n,3\right)=0\right\}}-\frac{8}{9} \chi_{\left\{\mathit{modp} \left(n,3\right)=1\right\}}.
\end{dmath}

In these outputs, $\chi_{\left\{\mathit{modp} \left(n ,m\right)=j\right\}}$ denotes the indicator function for the set of integers with non-negative remainder $j\in\{0,\ldots,m-1\}$ in their division by $m$; and $\chi_{\left\{\mathit{modp} \left(n ,1\right)=0\right\}}=1$.

\section{Structure and properties}

We consider sequences in $\KK^{\NN}$, where $\KK$ is a field of characteristic zero. In general, $\KK$ is a number field or a field such that the field extension $\KK/\QQ$ has a finite transcendence degree. We will use $n\in\NN$ as the index variable. All our results may certainly extend to negative indices -- $n\in\ZZ$, but we restrict ourselves to $\NN$ to fix the starting index at $0$ or $n_0\in\NN$ for factorial-like sequences. As it turns out, this is enough to introduce all the necessary concepts. When referring to an arbitrary sequence in $\KK^{\NN}$, we will denote that sequence with parentheses as $(s)_n\in\KK^{\NN}$ or simply $(s)$. The $n$th term of a sequence $(s)_n$, also called its general term, is $s_n$ (without the parentheses) or $s(n)$. We usually use the former notation; we use the latter when the subscript is already occupied. For instance, the term $T_1(n)$ is the $n$th term of the sequence $(T_1)_n$.

\subsection{Interlacements: $m$-fold indicator sequences}

Before diving into the concept of interlacement, let us recall a notion commonly used in set theory that will prove useful in the sequel. Let $\mathcal{A}$ be a set, and $A\subset \mathcal{A}$. The indicator function of $A$, denoted $\chi_A$, is defined as
\begin{align}\label{Indfunc}
    \chi_{A}\colon \mathcal{A} &\longrightarrow \{0, 1\} \nonumber\\
                         a     &\mapsto \chi_A(a)\coloneqq\begin{cases}1 \quad \text{ if } a\in A\\ 0 \quad \text{ otherwise }\end{cases}.
\end{align}

\begin{definition}[$m$-fold indicator sequence]\label{def:defmfoldind} A sequence $(s)_n$ is said to be $m$-fold indicator if there exists a positive integer $m_0$ such that $(s)_n$ is the indicator function of the non-negative integers in a coset of $\ZZ/m_0\ZZ$. In this case, $(s)_n$ is called $m_0$-fold indicator sequence, and $m_0$ is its characteristic.
\end{definition}
Note that the definition naturally extends to the whole set of integers $\ZZ$ but then requires a certain care for the starting index. Without a specific value, the terminology $m$-fold indicator sequence refers to any such sequence. In that case, $m$ may also be used as the corresponding characteristic without ambiguity.
\begin{example}\label{ex:defmfoldind}\item
\begin{itemize}
    \item $\chi_{\{2n+1,~ n\in\NN\}}:$ the indicator function of the set $\{2n+1,n\in\NN\}$ of odd natural numbers is a $2$-fold indicator sequence.
    \item $\chi_{\{3n+2,~ n\in\NN\}}$ is an $m$-fold indicator sequence of characteristic $3$. \QEDA 
\end{itemize} 
\end{example}
The following proposition shows that the characteristic of an $m$-fold indicator sequence is unique. 
\begin{proposition}\label{prop:uniqm} Let $m$ be a positive integer such that the sequence $(s)_n$ is $m$-fold indicator. Then $m$ is unique.
\end{proposition}

\vspace{-0.3cm}

\begin{proof} Denote by $[j]_m, j\in\{0,\ldots,m-1\}$, the non-negative integers of the coset in $\ZZ/m\ZZ$ of integers with remainder $j\geq 0$ in their division by $m$. Let $m_1$ and $m_2$ be two distinct positive integers such that the sequence $(s)_n$ is $m_1$-fold indicator and $m_2$-fold indicator. To obtain a contradiction, we only have to find an index $n$ at which $s_n=0$ and $s_n=1$. Without loss of generality, we can assume that $m_1<m_2$. Suppose that $(s)_n$ is the indicator sequence of $[j_1]_{m_1}$ and $[j_2]_{m_2}$, $j_1\in\{0,\ldots,m_1-1\}$, $j_2\in\{0,\ldots,m_2-1\}$. If $j_1\neq j_2$ then we are done since $s_{j_1}=1$ as an $m_1$-fold indicator sequence but $s_{j_1}=0$ as an $m_2$-fold indicator sequence. If $j_1=j_2$, then $j_1+m_1\in [j_1]_{m_1}$, but $j_1+m_1\notin [j_2]_{m_2}$. Hence, taking $n$ as $j_1+m_1$ yields a contradiction. Therefore we must have $m_1=m_2$.\CQFD
\end{proof}

A direct consequence of \Cref{prop:uniqm} is that we can count the number $m$-fold indicator sequences of characteristic $m$. The following corollary is also used as a definition to uniquely identify $m$-fold indicator sequences and fix a notation that we will use in the rest of the paper.
\begin{corollary}\label{cor:defmfremainder} There are exactly $m$ $m$-fold indicator sequences of characteristic $m$. For any $j\in\{0,\ldots,m-1\}$, the $m$-fold indicator sequence of characteristic $m$ and remainder $j$, denoted $(\chi_{\{j\mmod m\}})_n$, is defined by the general term
\begin{equation}
    \chi_{\{j\mmod m\}}(n) = \chi_{\{n\equiv j\mmod m\}} = \begin{cases}1 \,\quad \text{ if } n\equiv j\mmod m\\ 0 \,\quad \text{otherwise}\end{cases}, \quad n\in\NN,
\end{equation}
where $n\equiv j\mmod m$ means that $j$ is the non-negative remainder of $n$ in its division by $m$.
\end{corollary}
We mention that $m$-fold indicator sequences are intrinsically discussed in \cite{abramovmsparse}, but the attention there is on solving differential equations rather than studying the object itself. Here, we present definitions and properties to capture the mathematical essence of the concept of interlacement, which we will use later to define hypergeometric-type sequences. Our notation is closer to that used in \cite[Section 2.2, page 108]{KauersDfinite}, where $\chi_{\{n\equiv j\mmod m\}}$ is written as $\delta_{n\mmod m, j}$. However, there, $m$-fold indicator sequences are only considered when they have the same characteristic, in which case, their sums and products are straightforward to deduce.

We will now see what happens when we add and multiply $m$-fold indicator sequences of arbitrary characteristics and remainders. Since $0$ is not a divisor of any integer and that $1$ divides them all, we conventionally note $(\chi_{\{0\mmod 0\}})_n\coloneqq (\chi_{\{\mmod 0\}})_n$ the zero sequence $(0,0,\ldots)$, and $(\chi_{\{0\mmod 1\}})_n\coloneqq(\chi_{\{\mmod 1\}})_n$ the one sequence $(1,1,\ldots)$. These conventions make the next statements of this subsection more precise.

\begin{proposition}[Sum of $m$-fold indicator sequences]\label{prop:summfoldind} The sum of two non-zero $m$-fold indicator sequences is not an $m$-fold indicator sequence.
\end{proposition}
\begin{proof} If there are indices where the two $m$-fold indicator sequences take the same value $1$, then their sum takes the value $2$ at those indices and thus cannot be an $m$-fold indicator sequence. We could stop here as this first case necessarily happens for two $m$-fold indicator sequences. On the other hand, assuming that there are no indices where the two $m$-fold indicator sequences take the value $1$ and that it is $m$-fold indicator would imply that their sum is an $m$-fold indicator sequence with two characteristics, which is impossible according to \Cref{prop:uniqm}. \CQFD
\end{proof}
For products of $m$-fold indicator sequences, let us first look at some illustrative examples.

\begin{example}[Product of $m$-fold indicator sequences, part I]\label{ex:prodmfold1} Let us examine the product $(\chi_{\{1\mmod 4\}})_n\cdot (\chi_{\{1\mmod 6\}})_n$. The following table presents some of the first indices where the terms of both sequences are $1$, with their coincidences colored in red.

\vspace{-0.3cm}	

\begin{table}[ht]
\centering
\begin{tblr}{
  column{2} = {fg=red},
  column{6} = {fg=red},
  column{10} = {fg=red},
  column{14} = {fg=red},
}
$\mathbf{n \equiv 1\mmod 4}$ & 1 & 5 &   & 9 & 13 & 17 &    & 21 & 25 & 29 &    & 33 & 37 \\
$\mathbf{n \equiv 1\mmod 6}$ & 1 &   & 7 &   & 13 &    & 19 &    & 25 &    & 31 &    & 37 
\end{tblr}.
\caption{Indices where $\mfoldind{1}{4}=1$ and $\mfoldind{1}{6}=1$ for $n\leq 37$. The red color is for the coincidences.}
\label{tab:tab1}
\end{table}

\vspace{-0.4cm}

From \Cref{tab:tab1} one observes that the term $\chi_{\{n\equiv 1\mmod 4\}} \cdot \chi_{\{n\equiv 1\mmod 6\}}$ is apparently $1$ periodically. The period corresponds to the arithmetic progression $12n+1=\lcm(4,6)\cdot n+1$, where $\lcm$ stands for least common multiple. \QEDA
\end{example}

It turns out that the observation in \Cref{ex:prodmfold1} hides a general fact partly established by the following lemma.
\begin{lemma}\label{lem:prodlem1} Let $m_1,m_2\in\NN\setminus \{0\}$, and $j\in\{0,1,\ldots,\min\{m_1,m_2\}-1\}$. Then
\begin{equation}\label{eq:sameremaind}
    \mfoldind{j}{m_1} \cdot \chi_{\{n\equiv j \mmod m_2\}} = \chi_{\{n\equiv j \mmod \lcm(m_1,m_2)\}}.
\end{equation}
In other words, the product of two $m$-fold indicator sequences of the same remainder is an $m$-fold indicator sequence of this same remainder.
\end{lemma}
\begin{proof} $(\mfoldind[]{j}{m_1})_n$ is the indicator sequence of $[j]_{m_1}$ and $(\mfoldind[]{j}{m_2})_n$ is the one of $[j]_{m_2}$ (this notation was introduced in \Cref{prop:uniqm}). Thus for all non-negative integers $n$, $\mfoldind{j}{m_1}\cdot \mfoldind{j}{m_2}$ is $1$ if and only if $n\equiv j \mmod m_1$ and $n\equiv j \mmod m_2$. This implies that there exist non-negative integers $k_1$ and $k_2$, such that $n=m_1\cdot k_1 + j = m_2\cdot k_2 + j$. So $m_1\cdot k_1 = m_2\cdot k_2$. Thus $m_1 \mid m_2\cdot k_1$ ($m_1$ divides $m_2\cdot k_1$) and $m_2 \mid m_1\cdot k_1$. Let $\mu=\lcm(m_1,m_2)$. Then $\mu \mid m_1\cdot k_1$ and $\mu \mid m_2\cdot k_2$ since both $m_1$ and $m_2$ divide $m_1\cdot k_1$ and $m_2\cdot k_2$. Therefore $n\equiv j \mmod \mu$ and we conclude that $\mfoldind{j}{m_1}\cdot\mfoldind{j}{m_2}=\mfoldind{j}{\mu}$ by the uniqueness property (see \Cref{prop:uniqm}).
\CQFD
\end{proof}
Thus, we are now certain that $\mfoldind{1}{4}\cdot\mfoldind{1}{6}=\mfoldind{1}{12}$ for all $n\in\NN$. Let us consider the case of different remainders and generalize \Cref{lem:prodlem1}.
\begin{example}[Product of $m$-fold indicator sequences, part II]\label{ex:prodmfold2} We consider the product $(\chi_{\{1\mmod 4\}})_n\cdot (\chi_{\{2\mmod 6\}})_n$. The following table presents some of the first indices where the terms of both sequences are $1$, showing no coincidence for $n\leq 38$.
	
\vspace{-0.5cm}

\begin{table}
\centering
\begin{tblr}{}
$\mathbf{n \equiv 1\mmod 4}$ & 1 &   & 5 &   & 9 & 13 &     & 17 &    & 21 & 25 &    & 29 &    & 33 & 37 &  \\
$\mathbf{n \equiv 2\mmod 6}$ &   & 2 &   & 8 &   &    & 14  &    & 20 &    &    & 26 &    & 32 &    &    & 38 
\end{tblr}.
\caption{Indices where $\mfoldind{1}{4}=1$ and $\mfoldind{2}{6}=1$ for $n\leq 38$. No coincidence occurs.}
\end{table}

\vspace{-0.5cm}

We claim that $\mfoldind{1}{4}\cdot \mfoldind{2}{6}=0=\mfoldind{0}{0}$ for all $n\in\NN$. \QEDA
\end{example}

\begin{example}[Product of $m$-fold indicator sequences, part III]\label{ex:prodmfold3} We want to find $(\mfoldind[]{3}{4})_n\cdot (\mfoldind[]{2}{5})_n$. We consider the terms of indices $n\leq 67$.
	
\vspace{-0.5cm}

\begin{table}[h!]
\centering
\begin{tblr}{
  column{4} = {fg=red},
  column{12} = {fg=red},
}
$\mathbf{n \equiv 3 \mmod 4}$  &    & 3 & 7 & 11 &    & 15 &    & 19 &    & 23 & 27 & 31 &    & 35 &    & 39 \\
$\mathbf{n \equiv 2 \mmod 5}$  & 2  &   & 7 &    & 12 &    & 17 &    & 22 &    & 27 &    & 32 &    & 37 &    
\end{tblr}
\begin{tblr}{
  column{4} = {fg=red},
  column{12} = {fg=red},
}
$\mathbf{n \equiv 3 \mmod 4}$ &    & 43  & 47 & 51 &    & 55 &    & 59 &     & 63 & 67\\
$\mathbf{n \equiv 2 \mmod 5}$ & 42 &     & 47 &    & 52 &    & 57 &    &  62 &    & 67
\end{tblr}.
\caption{Indices where $\mfoldind{3}{4}=1$ and $\mfoldind{2}{5}=1$ for $n\leq 67$. Four coincidences: $n=7$, $n=27$, $n=47$, and $n=67$.}
\end{table}

\vspace{-1cm}

Claim: $\mfoldind{3}{4}\cdot \mfoldind{2}{5}=\mfoldind{7}{\lcm(4,5)}=\mfoldind{7}{20}$ for all $n\in\NN$. \QEDA
\end{example}

Together with \Cref{lem:prodlem1}, the following statement play a crucial role in proving one of our main results. They establish that the set of $m$-fold indicator sequences is multiplicatively closed.

\begin{lemma}[Product of $m$-fold indicator sequences]\label{lem:prodmfold} The product of two distinct $m$-fold indicator sequences of distinct remainders is an $m$-fold indicator sequence.
\end{lemma}
\begin{proof} Let $\mfoldind[]{j_1}{m_1}$ and $\mfoldind[]{j_2}{m_2}$ be two $m$-fold indicator sequences such that $j_1\neq j_2$. Let $\mu=\lcm(m_1,m_2)$ and
\begin{dmath}\label{eq:mathcalN}
    \mathcal{N}\coloneqq \left\{j\in\NN:~ j<\mu~ \text{ and there exist}~j_1,j_2\in\NN,j_1<m_1,j_2<m_2,~ \begin{cases}j\equiv j_1 \mmod m_1,\\ j\equiv j_2 \mmod m_2.
    \end{cases}  \right\}.
\end{dmath}
We consider two cases: $\mathcal{N}=\emptyset$ and $\mathcal{N}\neq \emptyset$.
\begin{itemize}
    \item[Case 1:] $\mathcal{N} = \emptyset$. This means that for indices less that $\mu$, there is no coincidence of $1$ between $\mfoldind{j_1}{m_1}$ and $\mfoldind{j_2}{m_2}$. We show that when this happens, the coincidence will not occur, and thus, the corresponding product of $m$-fold indicator sequences is the zero sequence. Let $n\in\NN$ such that $n\equiv j_1\mmod m_1 \equiv j_2\mmod m_2$. So there exist $k_1,k_2\in\NN,$ such that $m_1\cdot k_1 + j_1 = m_2\cdot k_2 + j_2$. Since $\mathcal{N}=\emptyset$, there exist $k_3,j_3\in\NN$, $k_3>0$ and $j_3< \mu$ such that $n=\mu\cdot k_3 + j_3$. Thus we have
    \[ m_1\cdot k_1 + j_1 = m_2\cdot k_2 + j_2 = \mu\cdot k_3 + j_3. \]
    We obtain a contradiction since this implies $j_3\in\mathcal{N}$. Therefore if $\mathcal{N}=\emptyset$ then $\mfoldind{j_1}{m_1}\cdot \mfoldind{j_2}{m_2}=0$ for all $n\in\NN$.
    \item[Case 2:] $\mathcal{N}\neq \emptyset$. To prove that the product is an $m$-fold indicator sequence, we only need to show that $|\mathcal{N}|=1$, i.e., $\mathcal{N}$ has only one element. We proceed by contradiction. As a subset of $\NN$, $\mathcal{N}$ has a least element. Let $j_0$ be that element. Then, following the reasoning of the first case, we can find integers $k_1,k_2$ such that
    \[j_0 = m_1\cdot k_1 + j_1 = m_2\cdot k_2 + j_2.\]
    Let $j_0'\in\mathcal{N}, j_0'\neq j_0$. Then $j_0'>j_0$ and we can find $k_1'$ and $k_2'$ such that
    \[j_0' = m_1\cdot k_1' + j_1' = m_2\cdot k_2' + j_2'.\]
    Thus
    \begin{eqnarray*}
        j_0'-j_0 &=& m_1(k_1'-k_1) + j_1'-j_1\\
                 &=& m_2(k_2'-k_2) + j_2'-j_2.
    \end{eqnarray*}
    Hence $j_0'-j_0\in\mathcal{N}$. In fact, for every $l\in\NN$, if $j_0'-l\cdot j_0\geq 0$, then $j_0'-l\cdot j_0\in\mathcal{N}$. By Euclidean division, let us write $j_0'=j_0\cdot q + r$, $0\leq r <j_0$. Then $r=j_0'- j_0\cdot q \in \mathcal{N}$, contradicting the fact that $j_0$ is the smallest element in $\mathcal{N}$. Therefore if $\mathcal{N}\neq\emptyset$ then 
    \[\mfoldind{j_1}{m_1}\cdot \mfoldind{j_2}{m_2}=\mfoldind{j_0}{\mu}\] 
    for all $n\in\NN$, where $j_0$ is the unique element of $\mathcal{N}$.
\end{itemize}
\CQFD
\end{proof}
One can also prove \Cref{lem:prodmfold} using the \textit{Chinese Remainder Theorem}. The proof of \Cref{lem:prodmfold} is constructive and gives an algorithmic way to find products of $m$-fold indicator sequences. Remark that they form a multiplicative group where the one sequence $\mfoldind[]{0}{1}$ is the unit element.
\begin{example}\item
\begin{itemize}
    \item $(\mfoldind[]{1}{4})_n\cdot (\mfoldind[]{2}{6})_n = (\mfoldind[]{}{0})_n=0$.
    \item $(\mfoldind[]{3}{4})_n\cdot (\mfoldind[]{2}{5})_n = (\mfoldind[]{7}{20})_n$.
    \item $(\mfoldind[]{1}{2})_n\cdot (\mfoldind[]{1}{3})_n = (\mfoldind[]{1}{6})_n=0$.
\end{itemize}
\QEDA
\end{example}

\subsection{Definition and structure}

We are now ready to define hypergeometric-type sequences. The idea is to encompass every possible linear combination of interlaced hypergeometric terms we can think of.
\begin{definition}[Hypergeometric-type sequence]\label{def:hyptypseq} A sequence $(s)_n$ is said to be of hypergeometric type if there exist finitely many $m$-fold indicator sequences $(\mfoldind[]{j_1}{m_1})_n,\ldots,(\mfoldind[]{j_l}{m_l})_n$ such that its general term $s_n$ writes
\begin{equation}\label{eq:hyptypseq}
    s_n = H_1(\sigma_1(n))\cdot \mfoldind{j_1}{m_1} + H_2(\sigma_2(n))\cdot \mfoldind{j_2}{m_2} + \cdots + H_l(\sigma_l(n))\cdot \mfoldind{j_l}{m_l},
\end{equation}
where $\sigma_i\colon \NN \longrightarrow \QQ$ is such that $\sigma_i(m_i\cdot n + j_i)\in\NN$, and $H_i(n)$ is a $\KK$-linear combination of hypergeometric terms, $i=1,\ldots,l$. We call the $H_i$'s the coefficients of $s_n$ (or $(s)_n$).
\end{definition}
\begin{remark}\item 
\begin{itemize}
    \item In \Cref{def:hyptypseq}, we used $H_i(\sigma(n))$ instead of ${H_i}_{\sigma(n)}$ to ease the notation and avoid confusion with indices.
    \item The specification of ``finitely many'' is mainly considered for algorithmic computation, though it seems unfeasible to envision arithmetic operations when the sum in \eqref{eq:hyptypseq} is infinite.    
    \item When one of the $m$-fold indicator terms in \eqref{eq:hyptypseq} is $\mfoldind[]{0}{1}$, the corresponding summand is replaced by its coefficient.
\end{itemize}
\end{remark}

Let $(\mathcal{H}_T)$ denote the set of hypergeometric-type sequences, and $\mathcal{H}_T$ be the set of their general terms. Unless otherwise stated, we assume that if $(s)_n\in(\mathcal{H}_T)$ then $s_n\in\mathcal{H}_T$. Every hypergeometric sequence is of hypergeometric type; in particular, every polynomial and rational sequence is of hypergeometric type.
\begin{example}[An example from the OEIS \cite{sloane2003line}]\label{ex:oeis1} The general term, say $a_n$, of the sequence \href{https://oeis.org/A307717}{A307717} counts the number of palindromic squares, $n^2$, of length $n$ (in the decimal basis) such that $n$ is also palindromic. Its explicit formula (see for instance \cite{kauers2022guessing,teguia2022fps}) is given by
\begin{equation}
    a_n \coloneqq \begin{cases} 0\phantom{195 + 203n- n^2}\quad \text{ if } n\equiv 0 \mmod 2\\
                              \frac{195 + 203n - 15n^2 + n^3}{192} \quad \text{ if } n\equiv 1 \mmod 4\\
                              \frac{501 + 107n - 9n^2 + n^3}{384}\phantom{5}  \quad \text{ if } n\equiv 3 \mmod 4
                    \end{cases}.
\end{equation}
The sequence $(a)_n$ is, of course, of hypergeometric type since its general term can be written as
\begin{equation}\label{eq:oeis1}
a_n = \frac{195 + 203n - 15n^2 + n^3}{192} \mfoldind{1}{4} + \frac{501 + 107n - 9n^2 + n^3}{384} \mfoldind{3}{4}.
\end{equation}
\QEDA
\end{example}

\begin{example}\label{ex:sumprod1} Let us consider two hypergeometric terms $h_n$ and $g_n$. By using subsequences of $(h)_n$ and $(g)_n$, we can construct several hypergeometric-type sequences. For instance, the two general terms
\begin{align}
    &u_n \coloneqq h_{(n-1)/4}\mfoldind{1}{4} + h_{(n-2)/5}\mfoldind{2}{5},\\
    &v_n \coloneqq g_{(n-2)/6}\mfoldind{2}{6} + g_{(n-3)/4}\mfoldind{3}{4},
\end{align}
are of hypergeometric type. Moreover, by definition, their sum
\begin{dmath}\label{eq:sumhyp1}
    u_n+v_n = h_{(n-1)/4}\mfoldind{1}{4} + g_{(n-3)/4}\mfoldind{3}{4} + h_{(n-2)/5}\mfoldind{2}{5} + g_{(n-2)/6}\mfoldind{2}{6},
\end{dmath}
is also of hypergeometric type. For their product, using \Cref{lem:prodmfold}, one can easily show that
\begin{dmath}\label{eq:prodhyp1}
    u_nv_n = h_{(n-2)/5} g_{(n-3)/4} \mfoldind{7}{20}.
\end{dmath}
Thus their product is also of hypergeometric type since the product of hypergeometric terms is a hypergeometric term. 
\QEDA
\end{example}
The previous example illustrates a general fact concerning the closure properties of hypergeometric-type sequences. We establish it in the following theorem.
\begin{theorem}\label{th:ringtheo} The set $(\mathcal{H}_T)$ of hypergeometric-type sequences is a ring.
\end{theorem}
\begin{proof}
Let $(s)_n,(s')_n\in(\mathcal{H}_T)$, and denote by $\mathfrak{M}$ and $\mathfrak{M}'$ the sets of $m$-fold indicator terms occurring in $s_n$ and $s_n'$, respectively. So we have
\begin{align*}
    &s_n = \sum_{\chi_i\in\mathfrak{M}} H_i(\sigma_i(n)) \chi_i(n),\\
    &s_n'= \sum_{\chi_i'\in\mathfrak{M}'} H_i'(\sigma_i'(n)) \chi_i'(n).
\end{align*}
Then, we can write
\begin{dmath}\label{eq:sumproof}
    s_n+s_n'= \sum_{\chi_i\in\mathfrak{M}\cap \mathfrak{M}'} \left(H_i(\sigma_i(n)) +  H_i'(\sigma_i'(n))\right)\chi_i(n) + \sum_{\chi_i\in\mathfrak{M}\setminus \mathfrak{M}'} H_i(\sigma_i(n)) \chi_i(n) + \sum_{\chi_i'\in\mathfrak{M}'\setminus \mathfrak{M}} H_i'(\sigma_i'(n)) \chi_i'(n) \in\mathcal{H}_T.
\end{dmath}
Hence $(s+s')_n\in(\mathcal{H}_T)$.

By distributivity of the multiplication with respect to addition, the product $s_n\cdot s_n'$ yields a sum of terms of the form
\begin{equation}\label{eq:prodproof}
    H_i(\sigma_i(n))\cdot H_j'(\sigma_j'(n)) \cdot \chi_i(n)\cdot\chi_j'(n),\quad \chi_i\in\mathfrak{M}, \chi_j'\in\mathfrak{M}'.
\end{equation}
Since the product of hypergeometric terms is a hypergeometric term, and that from \Cref{lem:prodlem1} and \Cref{lem:prodmfold} we know that $\chi_i(n)\cdot\chi_j'(n)$ is an $m$-fold indicator term, say $\chi_{i,j}''$, we deduce that \eqref{eq:prodproof} can be written as $H_{i,j}''(\sigma_{i,j}''(n)) \chi_{i,j}''(n)$, where $H_{i,j}''$ is a linear combination of hypergeometric terms. Hence $(s\cdot s')_n\in(\mathcal{H}_T)$. 

In conclusion $(\mathcal{H}_T)$ is a subring of $\KK^{\NN}$.
\CQFD
\end{proof}

\subsection{Generating functions}

With their LLL-based technique of guessing, Kauers and Koutschan were able to find a $6$th-order holonomic recurrence equation of degree $9$ for the OEIS sequence \href{https://oeis.org/A307717}{A307717} of \Cref{ex:oeis1} from its first $70$ terms. We have seen that this sequence is of hypergeometric type. We want to prove that such an equation always exists for any hypergeometric-type sequence.
\begin{proposition}\label{prop:holonomicity} Every hypergeometric-type sequence is P-recursive.
\end{proposition}
\begin{proof} We give sufficient arguments that show how to construct a holonomic recurrence equation satisfied by a given hypergeometric-type term. A basic example is given in \Cref{sec:holo}. Let $(s)_n\in(\mathcal{H}_T)$ such that
\[s_n\coloneqq \sum_{i=1}^l H_i(\sigma_i(n))\mfoldind{j_i}{m_i}.\]
Then for all $n\in\NN$, $1\leq i \leq l$, $s_{m_i n + j_i}=H_i(\sigma_i(m_i n + j_i))+\epsilon_i(n)$. We neglect $\epsilon_i(n)$ and look for a recurrence equation for $u_i(m_i n + j_i)\coloneqq H_i(\sigma_i(m_i n + j_i))$. Let $l_i$ be the number of $m_i$-fold hypergeometric terms in $u_i(m_i n + j_i)$. Since $u_i(m_i n + j_i)$ is a linear combination of $m_i$-fold hypergeometric terms, a recurrence equation of order at most $m_i\cdot l_i$ can be computed (see, e.g., \cite[Section 2]{BTvariant}). This yields a recurrence equation for the index $m_i n + j_i$. To obtain a recurrence in $n$, one substitutes $n$ by $(n-j_i)/m_i$ in the equation. Let us denote by $r_i$ the order of the resulting holonomic equation.

Notice that for all $1\leq i \leq l$, any $m_j$-fold hypergeometric terms in $\epsilon_i(n)$, $j\neq i$, is considered in one of the $u_k(n)$, $k=1,\ldots,l$. Therefore the span of all $m$-fold hypergeometric terms in $s_n$ is fully covered by the solution space of the system defined by the $l$ constructed holonomic recurrence equations.

Finally, using the addition closure property of P-recursive sequences, one can compute a holonomic equation of order at most $r_1+\cdots+r_l$ (see \cite{salvy1994gfun,mallinger1996algorithmic,koutschan2014holonomic,kauers2015ore}) satisfied by $s_n$. Hence $(s)_n$ is P-recursive. \CQFD
\end{proof}
From \Cref{prop:holonomicity}, we can say that the generating functions of hypergeometric-type sequences are D-finite functions (\cite{KauersDfinite}). However, we can be more specific. From \Cref{def:hyptypseries}, we can establish a natural link between hypergeometric-type sequences and proper hypergeometric-type power series \cite{BTphd,BTWKsymbconv}.
\begin{example} Let us consider the proper hypergeometric-type power series $\cos(z) + \sin(z)$. We have
\begin{align}
    \cos(z) + \sin(z) &\coloneqq \sum_{n=0}^{\infty} \frac{(-1)^n}{(2n)!}z^{2n} + \sum_{n=0}^{\infty} \frac{(-1)^{2n+1}}{(2n+1)!}z^{2n+1}\\
                      &= \sum_{n=0}^{\infty} \left(\frac{(-1)^{n/2}}{n!}\mfoldind{0}{2} + \frac{(-1)^{(n-1)/2}}{n!}\mfoldind{1}{2}\right) z^n.
\end{align}
Thus, the general coefficient of the power series of $\cos(z)+\sin(z)$ is a hypergeometric-type term. \QEDA
\end{example}

\begin{proposition}\label{prop:seriestoseq} There is a one-to-one correspondence between proper hypergeometric-type power series and hypergeometric-type sequences.
\end{proposition}
\begin{proof} This is established by the following equality
\begin{equation}\label{eq:propcorresp}
    \sum_{i=0}^{m_j-1}\sum_{n=0}^{\infty} a_i(m_j\,n+i)\,z^{m_j\,n+i} = \sum_{n=0}^{\infty}\left(\sum_{i=0}^{m_j-1} a_i(n) \mfoldind{i}{m_j}\right) z^n,
\end{equation}
which naturally links \Cref{def:hyptypseries} and \Cref{def:hyptypseq}. 
\CQFD
\end{proof}
In other words, the generating functions of hypergeometric-type sequences are hypergeometric-type functions. Thanks to the correspondence in \Cref{prop:seriestoseq}, we can relate hypergeometric-type sequences to $\KK(x)$-linear combinations of some special functions such as Bessel functions, Airy functions, and trigonometric functions. 

\section{Algorithmic approach}\label{sec:sec2}

Our motivation to study hypergeometric-type sequences came from the purpose of this section. Hypergeometric-type terms may be present in the sciences as trigonometric or elliptic functions with discrete arguments. A connection between hypergeometric terms and elliptic curves can be found in \cite[Section 1.3]{saito2013grobner}. Our purpose is to bring those representations into ``standard'' forms. Recent results in \cite{dreyfus2022hypertranscendence} show that the only meromorphic differentially algebraic functions (\cite{RSB2023,teguia2023arithmetic}) that are P-recursive are made with periodic functions and exponentials. It is thus reasonable to think of an algorithmic approach to convert the related discrete functions into hypergeometric-type normal forms. 

In this section, it is essential to understand the difference between an object $(1,-1,1,\ldots)\in(\mathcal{H}_T)$ and its term $(-1)^n\in\mathcal{H}_T$. The former represents the sequence $((-1)^n)_n$ as a mathematical object, and the latter as its `closed' mathematical writing. We recommend \cite[Chapter 3]{geddes1992algorithms} to remove any ambiguity in this formalism.

\subsection{Canonical and normal forms}

The definition of normal and canonical forms for hypergeometric-type sequences entails assessing computability in the ring $(\mathcal{H}_T)$. Indeed, an element of $(\mathcal{H}_T)$ may have many equivalent representations in $\mathcal{H}_T$.
\begin{example}[Distinct representations of same hypergeometric-type terms, part I]\label{ex:tworep} The sequence $(\sin^2(n\pi/4))$ is of hypergeometric type and has the two following representations:
\begin{align}
    \sin^2\left(\frac{n\pi}{4}\right) &=\frac{1}{2}\left(1-(-1)^{\frac{n}{2}} \mfoldind{0}{2}\right) \label{eq:canform1}\\
                                      &= \frac{1}{2}\left(1-\frac{I^n+(-I)^n}{2}\right)\label{eq:canform2}, 
\end{align}
where $I=\sqrt{-1}$ is the imaginary unit.
\QEDA
\end{example}
The idea of a canonical form is to have a unique and ``simple'' representation of a mathematical object in a certain class, here $\mathcal{H}_T$. Observe that the representation \eqref{eq:canform2} of $s_n\coloneqq \sin^2\left(n\pi/4\right)$ requires to work on $\QQ(I)$, whereas \eqref{eq:canform1} is a formula over $\QQ$ (no extension field). Thus, \eqref{eq:canform1} and \eqref{eq:canform2} are valid in $\QQ(I)$, making \eqref{eq:canform2} less appropriate as a canonical form of $s_n$. Hence writing formulas over the minimal field extension reduces the possible representations of a hypergeometric-type term. Unfortunately, uniqueness remains an issue even in base fields. We give two examples below.
\begin{example}[Distinct representations of same hypergeometric-type terms, part II]\label{ex:nocanform}\item 
\begin{enumerate}
    \item The general coefficient of $\cosh(z)$ has the two formulas: 
    \begin{align}
    \frac{1}{n!}\mfoldind{0}{2} &= \frac{1+(-1)^n}{2\cdot n!} \label{eq:canform3}.
    \end{align}
    \item The following was observed from different formulas of the general term of \href{https://oeis.org/A212579}{A212579}:
    \begin{align}
    \frac{31}{3} - \mfoldind{0}{2} &= \frac{1}{2}\left(\frac{59}{3} - (-1)^n \right). \label{eq:canform4}
\end{align}
\end{enumerate}
\QEDA
\end{example}
While for compactness reasons, one might prefer the left-hand side for the first item in \Cref{ex:nocanform}, both sides seem to have a relatively similar compactness for the second item. We may also choose to define a canonical form by eliminating all alternations in the formula. For instance, $(-1)^n$ can be written as $\mfoldind{0}{2}-\mfoldind{1}{2}$, the latter form being seen as canonical. In this view, our canonical forms would be the left-hand sides in \Cref{ex:nocanform}. For \Cref{ex:tworep}, \eqref{eq:canform1} may be further simplified since $(-1)^{n/2}\mfoldind{0}{2}$ alternates between $-\mfoldind{0}{2}$ and $\mfoldind{0}{2}$. So, a hypergeometric-type canonical form of $\sin^2\left(n\pi/4\right)$ would contain $4$-fold indicator terms, which sounds reasonable with the $4$ occurring in its expression. However, as the second item in \Cref{ex:nocanform} shows, further simplifications need to be done to reduce the number of $m$-fold indicator terms after substituting alternating elements. The fact that $\mfoldind{0}{2}$ survived allows us to think that there might be another way to write the formula with $\mfoldind{1}{2}$. Therefore, uniqueness may not still be guaranteed.

Why do we care about canonical form at all? The main reason is that it completely solves (theoretically) the zero-equivalence problem in $\mathcal{H}_T$. However, our motivation is to bring expressions not written as elements of $\mathcal{H}_T$ into easily recognizable hypergeometric-type terms whenever possible. The above discussion presents the difficulty of defining a canonical form in $\mathcal{H}_T$ and forces us to reduce ourselves to normal forms. We have already introduced them, but we give the definition below for formal reference.
\begin{definition}[A normal form in $\mathcal{H}_T$] Any representation of a hypergeometric-type term as in \eqref{eq:hyptypseq} is a normal form.
\end{definition}

Having stated our normal form, we need to give an algorithm for the normalization. Note that this still solves the zero equivalence problem by the unique representation of the zero sequence. The remaining part of the paper describes the algorithm sustaining this fact.

\subsection{Finding holonomic recurrence equations}\label{sec:holo}
This subsection addresses finding a holonomic recurrence equation satisfied by a given hypergeometric-type ``expression''. It means that the given formula is not necessarily in the form of \eqref{eq:hyptypseq}. The algorithm behind this conversion is the first step of our general algorithm toward finding hypergeometric-type normal forms. We mention that there is no fundamental result in this part of the paper because there are many existing software to compute univariate holonomic (differential and difference) equations. Some references are \cite{salvy1994gfun,mallinger1996algorithmic,koutschan2014holonomic,kauers2015ore,FPS}. However, unlike the differential case for which software packages are easily accessible for any computer algebra system, general-purpose algorithms for finding recurrences from holonomic expressions do not seem available in the difference case. For instance, the well-known \texttt{GFUN} package \cite{salvy1994gfun} misses such an implementation. That is one reason why we decided to include this subsection. We adapt the \texttt{HolonomicDE} algorithm (available within Maple 2022 as \texttt{DEtools:-FindODE}) of \cite{FPS} to the case of recurrences. For details on the original algorithm, see the explanation from \cite{koepf1992power}. Given a term $s_n$, the aim is to find $C_1,\ldots,C_N\in\QQ(n)$ such that
\[s_{n+N}+C_N\cdot s_{n+N-1} +\cdots + C_1\cdot s_n=0.\]
We consider an efficient variant of this method as proposed in \cite[Section 4.1.2]{BTphd},\cite[Section 2]{BTWKsymbconv}. The idea is to write $s_n$ and its $N$ first shifts in the same basis and solve the linear system that expresses their linear dependency over $\KK(n)$. We explain how the algorithm works in the following example.
\begin{example}\label{ex:recn1n} Let $s_n\coloneqq n! + \frac{1}{n!}$.
\begin{enumerate}
    \item $N=0$: since $\frac{n!}{1/n!}=n!^2\notin \QQ(n)$, we consider the basis $(e_1,e_2)$, where $e_1=n!$ and $e_2=1/n!$. Thus
                            \[s_n = e_1 + e_2,\]
and at this stage, the matrix $H$ of the components is $H=[1, 1]$. The rows of $H$ are the components of $s_{n+N}$ in the $(N+1)$st basis.
    \item $N=1$: 
    \[s_{n+1}=(n+1)! + \frac{1}{(n+1)!} = (n+1)\cdot e_1 + \frac{1}{(n+1)} \cdot e_2.\]
    \[H=\begin{bmatrix}1 & 1\\ n+1 & \frac{1}{n+1}\end{bmatrix}.\] 
    Since $s_{n+1}$ and $s_n$ are written in the same basis, we try to solve the system
    \[\begin{bmatrix}1\\ 1\end{bmatrix}\cdot C = \begin{bmatrix}-(n+1) \\ -\frac{1}{n+1}\end{bmatrix},~ C\in \QQ(n).\]
    The right-hand side is the negative transpose of the last row of $H$, and the matrix of the left-hand side is the transpose of the first $N$ rows of $H$. The obtained system has no solution, so we move on to the next iteration.
    \item  $N=2$:
    \[s_{n+2}= (n+2)(n+1)\cdot e_1 + \frac{1}{(n+2)(n+1)} \cdot e_2.\]
    \[H=\begin{bmatrix}1 & 1\\ n+1 & \frac{1}{n+1}\\ (n+2)(n+1) & \frac{1}{(n+2)(n+1)}\end{bmatrix}.\] 
    We solve the system
    \[\begin{bmatrix}1 & n+1\\ 1 & \frac{1}{n+1}\end{bmatrix}\cdot C = \begin{bmatrix}-(n+2)(n+1)\\ -\frac{1}{(n+2)(n+1)}\end{bmatrix},~ C\in \QQ(n)^2,\]
    and get a unique solution
    \[C=\begin{bmatrix}C_1\coloneqq\frac{(n+3)(n+1)^2}{n(n+2)^2}\\ C_2\coloneqq-\frac{(n^2+3n+1)(n^2+3n+3)}{n(n+2)^2}\end{bmatrix}.\]
    Thus $s_n$ satisfies the equation
    \[a_{n+2} + C_2 a_{n+1} + C_1 a_n =0.\]
    After clearing denominators, we get the holonomic recurrence equation
    \begin{equation}\label{eq:recn1n}
        (n+3)(n+1)^2a_n - (n^2+3n+1)(n^2+3n+3) a_{n+1} + n(n+2)^2 a_{n+2}=0,
    \end{equation}
    satisfied by $s_n$.
\end{enumerate}
\QEDA
\end{example}

\begin{example}\label{ex:recsin2} For $\sin^2(n\pi/4)$ the algorithm leads to the recurrence
\begin{equation}
    -a_n+a_{n+1}-a_{n+2}+a_{n+3}=0.
\end{equation}
Of course, the algorithm works in this case because the expansion formulas of trigonometric functions are used.
\QEDA
\end{example}

The above-outlined algorithm cannot apply to hypergeometric-type terms written in normal forms. For those terms we use the construction highlighted in the proof of \Cref{prop:holonomicity}. We give one basic example to illustrate how it works.
\begin{example} Let $s_n= \left(\frac{1}{3^{\frac{n}{2}}}+(-5)^{\frac{n}{2}}\right)\mfoldind{0}{2} + 2^{\frac{n}{3}} \mfoldind{0}{3}$. 

We consider $u_1(2n)=\frac{1}{3^n}+(-5)^n$ and $u_2(3n)=2^n$. To find a recurrence for $u_1(2n)$, we use the addition algorithm with $2$-shifts and find
\[ 5\,u_1(2n)-14\,u_1(2n+2)-3\,u_1(2n+4)=0. \]
Hence the equation for $u_1(n)$:
\[ 5\,u_1(n) - 14\,u_1(n+2)-3\,u_1(n+4)=0. \]
Similarly, $u_2(n)$ satisfies
\[ 2\, u_2(n) - u_2(n+3) = 0.\]
Finally, using the addition closure property for holonomic sequences we get the equation
\begin{equation}
 10\,a_{n} - 28\,a_{n + 2} - 5\,a_{n + 3} - 6\,a_{n+4} + 14\,a_{n+5} +  3\,a_{n + 7} = 0,
\end{equation}
satisfied by $s_n$.
\QEDA    
\end{example}

We will denote by HolonomicRE($s_n$,$a(n)$,$d$) the algorithm that applies the algorithm in the proof of \Cref{prop:holonomicity} if $s_n$ is already in normal form, i.e., $s_n\in\mathcal{H}_T$, and the algorithm outlined in \Cref{ex:recn1n} otherwise. The output is either a holonomic recurrence equation of order at most $d\in\NN$ in the indeterminate $a(n)$, or \texttt{FAIL} when such an equation is not found. We can omit $d$ for hypergeometric-type terms since the recurrence is obtained by construction and not by search. 

We mention that $\sin^2(z\pi/4)$ also satisfies a holonomic differential equation. The work in \cite{dreyfus2022hypertranscendence} suggests that we can use trigonometric functions to generate terms that satisfy holonomic recurrence equations. The remaining steps of our algorithm help to verify whether these terms are of hypergeometric type or not.

\subsection{Finding normal forms}

Let $(s)_n\in(\mathcal{H}_T)$ such that the given expression $s_n$ is not an element of $\mathcal{H}_T$. We want to find a representation of $s_n$ in $\mathcal{H}_T$. Suppose that $s_n$ is a solution of the following $d$th-order recurrence equation:
\begin{equation}\label{eq:REnform}
    P_d(n) a_{n+d} + P_{d-1}(n)a_{n+d-1}+\cdots+ P_0(n)a_n=0, 
\end{equation}
with polynomial coefficients $P_d,P_{d-1},\ldots,P_0\in\KK[x],$ $P_dP_0\neq 0$. Using mfoldHyper \cite[Section 3]{BTWKsymbconv}, we can compute a basis of $m$-fold hypergeometric term solutions of \eqref{eq:REnform}. This may be written as
\begin{equation}\label{eq:nformbasis}
    \mathfrak{B}\coloneqq \left\lbrace \left\lbrace m_i,\mathfrak{B}_i\right\rbrace, i=1,\ldots,N\right\rbrace \coloneqq  \left\lbrace \left\lbrace m_i,\left\lbrace h_{i,1}(m_in),\ldots,h_{i,l_i}(m_in)\right\rbrace\right\rbrace, i=1,\ldots,N\right\rbrace,
\end{equation}
$m_i\in\NN\setminus \lbrace 0 \rbrace$. For each $h_{i,j}\in \mathfrak{B}_i$, there are $m_i-1$ other solutions, namely $h_{i,j}(m_in+k_j)$, $k_j=1,~\ldots,m_i-1$. These other solutions can also be generated by mfoldHyper at the user's request. Note that the reason why the basis \eqref{eq:nformbasis} is written in this form is because the primary purpose of mfoldHyper is to find general coefficients of formal power series. We recall that mfoldHyper is an extension of the algorithms by Petkov\v{s}ek (Hyper) and van Hoeij \cite{petkovvsek1992hypergeometric,van1999finite,BTvariant}. Thanks to the correspondence of \Cref{prop:seriestoseq}, the output of mfoldHyper can be easily used to find a hypergeometric-type representation of $s_n$.

To obtain a hypergeometric-type formula for $s_n$, we look for constant coefficients $c_{i,j,k_j}\in\KK, i=1,\ldots,N,$ $j=1,\ldots,l_i$, $k_j=0,\ldots,m_i-1$, such that
\begin{equation}\label{eq:ansatz}
    s_n = \sum_{0\leq k_j\leq m_i-1, 1\leq j\leq l_i, 1\leq i\leq N} c_{i,j,k_j} h_{i,j,k_j}(n) \mfoldind{k_j}{m_i}.
\end{equation}
We evaluate both sides of \eqref{eq:ansatz} to obtain a Cramer system for the unknown $c_{i,j,k_j}$'s. Its solutions lead to a hypergeometric-type representation of $s_n$. We mention that the resulting system can have many solutions because some sub-bases or mixing of elements from the basis in \eqref{eq:nformbasis} may span the same vector space over different field extensions. Nevertheless, we need to select one of them to get the normal form we want. Let us give some examples.

\begin{example}[$s_n\coloneqq\sin^2\left(\frac{n\pi}{4}\right)$]\label{ex:sin2nform} As presented in \Cref{ex:recsin2}, $s_n$ satisfies the recurrence equation:
\[-a_n+a_{n+1}-a_{n+2}+a_{n+3}=0.\]
Algorithm mfoldHyper finds the following basis of solutions over $\QQ$:
\begin{equation}
    \left\{ \left\{1,\{1\}\right\}, \left\{2,\{(-1)^n\}\right\}\right\}.
\end{equation}
More solutions can be found if one enables computations over field extensions. This is avoided as much as possible to have the chance to obtain a normal form over the base field. We write
\[s_n = c_1 + c_2 (-1)^{\frac{n}{2}}\mfoldind{0}{2} + c_3 (-1)^{\frac{n-1}{2}}\mfoldind{1}{2},\]
and use the first terms $s_0,s_1,s_2$ to obtain the linear system
\[
    \begin{cases}
        c_1 + c_2 = 0\\
        c_1 + c_3 = \frac{1}{2}\\
        c_1 - c_2 = 1
    \end{cases}.
\]
The system has a unique solution which leads to the following normal form for $s_n$:
\begin{equation}
    s_n=\sin^2\left(\frac{n\pi}{4}\right) = \frac{1}{2} - \frac{(-1)^{\frac{n}{2}}}{2}\mfoldind{0}{2}
\end{equation}
\QEDA
\end{example}

\begin{example}[$s_n\coloneqq \sin\left(\frac{\pi}{6}\cos\left(n\pi\right)\right)\sin\left(\frac{n\pi}{4}\right)$]\label{ex:nformex2} The given term satisfies the recurrence equation:
\begin{equation}\label{eq:recnformex2}
    a_n + \sqrt{2} a_{n+1} + a_{n+2}=0.
\end{equation}
This equation does not have $m$-fold hypergeometric term solutions over $\QQ$, not even over $\QQ(\sqrt{2})$. Enabling extension fields allows mfoldHyper to find the basis of solutions
\begin{equation}\label{eq:b1nformex2}
    \left\{\left\{1, \left\{\left(\texttt{RootOf}\left(1+\sqrt{2} X + X^2\right)\right)^n\right\}\right\}\right\},
\end{equation}
where $\texttt{(RootOf(P(X)))}^n$ is a compact notation of $\alpha^n$, for all $\alpha, P(\alpha)=0$. Thus we have two hypergeometric terms over $\QQ(\sqrt{2}, I)$. The algorithm can proceed as in \Cref{ex:sin2nform} and find a hypergeometric-type representation of $s_n$. However, as this is not what our implementation will do (\Cref{sec:sec3}), we want to present a technique that our implementation does to avoid field extensions. This may also justify why we cannot always find normal forms over base fields. The point is, as designed, the algorithm of \Cref{sec:holo} tries to compute a holonomic recurrence equation of the smallest order. However, the least-order recurrence equation may not contain term solutions over the base field of the given hypergeometric-type sequence. Thus, it might be relevant to look for other recurrence equations. To do so, we ask the algorithm to search for a recurrence equation between $2$-shifts of $s_n$, i.e., $s_n, s_{n+2}, s_{n+4},$ $\ldots$. We obtain the two-term recurrence equation
\begin{equation}\label{eq:rec2nformex2}
    a_n + a_{n+4} = 0.
\end{equation}
Hence, the basis of term solutions over $\QQ$:
\begin{equation}
    \left\{\left\{4,\left\{\left(-1\right)^{n}\right\}\right\}\right\}.
\end{equation}
At this stage, we are sure to obtain a normal form in the corresponding base field since any algebraic number in the formula will come from the evaluation of $s_n$. For the ansatz
\begin{dmath*}
    s_n = c_0 (-1)^{\frac{n}{4}}\mfoldind{0}{4} +  c_1 (-1)^{\frac{n-1}{4}}\mfoldind{1}{4} +  c_2 (-1)^{\frac{n-2}{4}}\mfoldind{2}{4} +  c_3 (-1)^{\frac{n-3}{4}}\mfoldind{3}{4},
\end{dmath*}
we get the linear system:
\[
\begin{cases}
c_0=0\\
c_1=-\frac{\sqrt{2}}{4}\\
c_2=\frac{1}{2}\\
c_3=-\frac{\sqrt{2}}{4}
\end{cases}.
\]
Therefore
\begin{dmath*}
    s_n = \sin\left(\frac{\pi}{6}\cos\left(n\pi\right)\right)\sin\left(\frac{n\pi}{4}\right) =  -\frac{(-1)^{\frac{n-1}{4}}\,\sqrt{2}}{2}\mfoldind{1}{4} +  \frac{1}{2} (-1)^{\frac{n-2}{4}}\mfoldind{2}{4} +  -\frac{(-1)^{\frac{n-3}{4}}\,\sqrt{2}}{4}\mfoldind{3}{4}
\end{dmath*}
\QEDA
\end{example}

Let us present all the steps of our algorithmic approach to detecting hypergeometric-type terms by writing them in normal forms.

\vspace{-0.25cm}

\begin{algorithm}[ht]\caption{ Finding hypergeometric type formulas }\label{algo:Algo1}
    \begin{algorithmic} 
    \\ \Require A general term $s_n$ of a sequence $(s)_n\in\KK^{\NN}$, and a positive integer $d$. If $s_n\in\mathcal{H}_T$, then $d$ may be computed as: the sum of (number of hypergeometric term in each coefficient) $\times$ (the corresponding characteristic).
    \Ensure Either 
    \begin{itemize}
        \item \texttt{FAIL}, meaning that ``no holonomic recurrence equation of order at most $d$ was found'';
        \item a holonomic recurrence equation of order at most $d$ with enough initial values to identify $(s)_n$ uniquely: this means that ``$(s)_n\notin (\mathcal{H}_T)$'';
        \item a hypergeometric-type normal form, meaning that ``$(s)_n\in(\mathcal{H}_T)$''.
    \end{itemize} 
    \begin{enumerate}
    \item Apply HolonomicRE$(s_n,a(n),d)$ (\Cref{sec:holo}) and call the result $RE$.
    \item If $RE=\texttt{FAIL}$ then stop and return it. //comment: $d$ may be small.
    \item \label{setp:retohts} $RE$ is a holonomic recurrence equation of order $r\leq d$. Use \texttt{mfoldHyper} to compute a basis of $m$-fold hypergeometric term solutions of $RE$ over $\KK$ and denote it $\mathfrak{B}$.
    \item If $\mathfrak{B}$ is empty then stop and return $RE$ together with $a_0=s_0,\ldots,a_{r-1}=s_{r-1}$.
    \item $\mathfrak{B}$ is not empty and has the form 
            \[\mathfrak{B}\coloneqq \left\lbrace \left\lbrace m_i,\left\lbrace h_{i,1}(m_in),\ldots,h_{i,l_i}(m_in)\right\rbrace\right\rbrace, i=1,\ldots,N\right\rbrace,\]
        as in \eqref{eq:nformbasis}. Let
        \[u_n \coloneqq \sum_{0\leq k_j\leq m_i-1, 1\leq j\leq l_i, 1\leq i\leq N} c_{i,j,k_j} h_{i,j,k_j}(n) \mfoldind{k_j}{m_i},\]
        as in \eqref{eq:ansatz}, with the unknown constants $c_{i,j,k_j}\in\KK, i=1,\ldots,N,$ $j=1,\ldots,l_i$, $k_j=0,\ldots,m_i-1$.
    \item\label{step:p} Let $p$ be the number of constant $c_{i,j,k_j}$. $p=\sum_{i=1}^N m_i\cdot l_i$
    \item\label{step:E0} Let $E_0$ be a finite set of non-negative integers that are not roots of the leading and the trailing polynomial coefficients of $RE$, such that $E_0$ evaluates $u_n=s_n$ to a linear system of rank at least $p$.
    \item Solve the linear system $u_j=s_j, j\in E_0$, and let $S$ be the set of solutions.
    \item \label{step:solsys} If $S$ is empty then stop and return $RE$ together with $a_0=s_0,\ldots,a_{r-1}=s_{r-1}$.
    \item Return the substitution of a solution in $S$ into $u_n$.
    \end{enumerate}	
    \end{algorithmic}
\end{algorithm}

\vspace{-0.75cm}

\begin{remark}\label{rem:rem1}\item 
\begin{itemize}
    \item Note that we omitted the steps where we try to avoid field extensions to simplify the algorithm. The idea is to use mfoldHyper over $\QQ$ in step \ref{setp:retohts} with a few more recurrences satisfied by $s_n$, and see if it leads to a non-empty $S$ at step \ref{step:solsys}.
    \item The reason for avoiding roots of the leading and the trailing coefficients is a singularity issue. See the discussion in \cite[Section 2.2]{KauersDfinite}.
    \item The set $E_0$ in step \ref{step:E0} can be chosen as $\{n_0,\ldots,n_0+p\}$, where $n_0$ is a non-negative integer strictly greater than the maximum integer root of the leading and the trailing polynomial coefficients. However, finding $n_0$ that way may not be the best approach when symbolic values occur in the equation. One could look for such integer intervals by evaluation at consecutive indices starting from $0$. The latter approach would be inappropriate in today's computer only if the recurrence equation has thousands or billions of $m$-fold hypergeometric term solutions.
\end{itemize}
\end{remark}

\begin{theorem}\label{th:theo2} \Cref{algo:Algo1} is correct.
\end{theorem}

\Cref{th:theo2} is deduced from the previous paragraphs and \Cref{rem:rem1}. \Cref{algo:Algo1} is a transformation for finding normal forms of hypergeometric-type terms for which holonomic recurrence equations are found in its first step. The zero sequence may be returned as $0$ or as the zeroth-order holonomic recurrence equation. Thus, we can identify distinct hypergeometric-type terms.

\section{Implementation}\label{sec:sec3}

We implemented \Cref{algo:Algo1} with Maple as a command in the package \texttt{HyperTypeSeq} \cite{HyperTypeSeq}. The package currently contains three commands: \texttt{HolonomicRE}, \texttt{REtoHTS}, and \texttt{HTS}.
\begin{enumerate}
    \item \texttt{HolonomicRE} adapts \texttt{HolonomicDE} from \texttt{FPS} \cite{FPS} to search for a holonomic recurrence equation from an expression and a given bound. The syntax is 
    \[\texttt{HolonomicRE}(\texttt{expr},\texttt{a}(\texttt{n}),\texttt{maxreorder}=\texttt{d},\texttt{reshift}=\texttt{t}),\]
    where \texttt{maxreorder} and \texttt{reshift} are optional with default values $10$ and $1$, respectively. \texttt{expr} is a term in \texttt{n}, and \texttt{a} is the name of the unknown for the equation. \texttt{maxreorder} is the maximum order of the holonomic recurrence equation sought, and \texttt{reshift} is the minimal possible shift of $\texttt{a}(\texttt{n})$ in the recurrence equation sought. The current version of \texttt{HolonomicRE} still misses an implementation for finding recurrence equations from hypergeometric-type terms containing $m$-fold indicator terms.
    \item \texttt{REtoHTS} applies \Cref{algo:Algo1} from step \ref{setp:retohts}. The syntax is
    \[\texttt{REtoHTS}(\texttt{RE},\texttt{a}(\texttt{n}),\texttt{P}).\]
    \texttt{RE} is the holonomic recurrence equation and \texttt{a}(\texttt{n}) is the unknown term in it. \texttt{P} is a procedure for computing values of the sequence at any index. \texttt{P} can also be a list of initial values; however, the list must contain the values of the evaluations of \texttt{expr} starting from $0$. 
    
    With finding holonomic recurrence equations for sequences in enumerative combinatorics, \texttt{REtoHTS} may be useful for finding new formulas.
    \item \texttt{HTS} implements \Cref{algo:Algo1} with the syntax
    \[\texttt{HTS}(\texttt{expr},\texttt{n}),\]
    with self-explanatory arguments from the previous commands. The argument \texttt{maxreorder} is also optional for \texttt{HTS}.
\end{enumerate}

For our implementation $\mfoldind{j}{m}=\chi_{\left\{\mathit{modp} \left(n ,m\right)=j\right\}}$. We can now present more sophisticated conversions of trigonometric expressions into hypergeometric-type terms. We encountered an implementation issue with Maple 2022 and Maple 2023; we could not obtain some formulas that Maple 2019 and Maple 2021 found within seconds with our code. Maple 2022 and Maple 2023 keep running. Simple checking on our implementation tells us that the problem comes from the linear system solver \texttt{SolveTools:-Linear}. We will see one of the expressions that led to this misbehavior. So, note that all the formulas in \Cref{ex:manyexpl} are obtained within seconds ($\leq 4$s) with Maple 2021. 

\begin{example}[Some expressions of hypergeometric type]\label{ex:manyexpl}\item 
    \begin{enumerate}        
        \item 
\begin{lstlisting}
> with(HyperTypeSeq):
> HTS(sin(Pi*cos(n*Pi)/6)*cos(n*Pi/4),n)
\end{lstlisting}
\begin{dmath}\label{eq:manyex1}
\frac{\left(-1\right)^{\frac{n}{4}} \chi_{\left\{\mathit{modp} \left(n ,4\right)=0\right\}}}{2}-\frac{\sqrt{2}\, \left(-1\right)^{\frac{n}{4}-\frac{1}{4}} \chi_{\left\{\mathit{modp} \left(n ,4\right)=1\right\}}}{4}+\frac{\sqrt{2}\, \left(-1\right)^{\frac{n}{4}-\frac{3}{4}} \chi_{\left\{\mathit{modp} \left(n ,4\right)=3\right\}}}{4}
\end{dmath}
        
        \item 
\begin{lstlisting}
> HTS(sin(cos(n*Pi/3)*Pi),n)
\end{lstlisting}
\begin{dmath}\label{eq:manyex2}
\left(-1\right)^{\frac{n}{3}-\frac{1}{3}} \chi_{\left\{\mathit{modp} \left(n ,3\right)=1\right\}}-\left(-1\right)^{\frac{n}{3}-\frac{2}{3}} \chi_{\left\{\mathit{modp} \left(n ,3\right)=2\right\}}
\end{dmath}

        \item 
\begin{lstlisting}
> HTS(tan(n*Pi/4),n)
\end{lstlisting}
\begin{dmath}\label{eq:manyex3}
\chi_{\left\{\mathit{modp} \left(n ,4\right)=1\right\}}+\left(\munderset{n \rightarrow 2}{\textcolor{gray}{\mathrm{lim}}}\! \tan \! \left(\frac{n \pi}{4}\right)\right) \chi_{\left\{\mathit{modp} \left(n ,4\right)=2\right\}}-\chi_{\left\{\mathit{modp} \left(n ,4\right)=3\right\}}
\end{dmath}

        \item 
\begin{lstlisting}
> HTS(tan(n*Pi/3),n)
\end{lstlisting}
\begin{dmath}\label{eq:manyex4}
\sqrt{3}\, \chi_{\left\{\mathit{modp} \left(n ,3\right)=1\right\}}-\sqrt{3}\, \chi_{\left\{\mathit{modp} \left(n ,3\right)=2\right\}}
\end{dmath}

        \item Chebyshev polynomials:
\begin{lstlisting}
> HTS(cos(n*arccos(x)),n)
\end{lstlisting}
\begin{dmath}\label{eq:manyex6}
\frac{\left(x -\sqrt{x^{2}-1}\right)^{n}}{2}+\frac{\left(x +\sqrt{x^{2}-1}\right)^{n}}{2}
\end{dmath}

        \item 
\begin{lstlisting}
> HTS(sin(n*Pi/6)*cos(n*Pi/3)-sin(n*Pi/2),n)
\end{lstlisting}
\begin{dmath}\label{eq:manyex7}
-\frac{\mathrm{I} \left(\frac{\sqrt{3}}{2}-\frac{\mathrm{I}}{2}\right)^{n}}{4}+\frac{\mathrm{I} \left(\frac{\sqrt{3}}{2}+\frac{\mathrm{I}}{2}\right)^{n}}{4}-\frac{\left(-1\right)^{\frac{n}{2}-\frac{1}{2}} \chi_{\left\{\mathit{modp} \left(n ,2\right)=1\right\}}}{2}
\end{dmath}

        \item 
\begin{lstlisting}
> HTS(sin(n*Pi/4)^2*cos(n*Pi/6)^2,n)
\end{lstlisting}
\begin{dmath}\label{eq:manyex8}
\frac{1}{4}-\frac{\left(\mathrm{-I}\right)^{n}}{8}+\frac{\left(\frac{1}{2}-\frac{\mathrm{I} \sqrt{3}}{2}\right)^{n}}{8}+\frac{\left(\frac{1}{2}+\frac{\mathrm{I} \sqrt{3}}{2}\right)^{n}}{8}-\frac{\mathrm{I} \left(-1\right)^{\frac{n}{2}-\frac{1}{2}} \chi_{\left\{\mathit{modp} \left(n ,2\right)=1\right\}}}{8}-\frac{3 \left(\mathrm{-I}\right)^{\frac{n}{3}} \chi_{\left\{\mathit{modp} \left(n ,3\right)=0\right\}}}{8}-\frac{3 \,\mathrm{I} \left(-1\right)^{\frac{n}{6}-\frac{1}{2}} \chi_{\left\{\mathit{modp} \left(n ,6\right)=3\right\}}}{8}
\end{dmath}

The following formula could not be obtained with Maple 2023 and Maple 2022. That is the reason why we used Maple 2021 for all the examples in \Cref{ex:manyexpl}.

        \item 
\begin{lstlisting}
> HTS(sin(n*Pi/4)^2*cos(n*Pi/6)^4,n,maxreorder=12)
\end{lstlisting}
\begin{dmath}\label{eq:manyex9}
\frac{9}{32}+\left(\left(\frac{3}{32}+\frac{\mathrm{I} \sqrt{3}}{8}\right) \left(-1\right)^{\frac{n}{2}}+\frac{\left(-\frac{1}{2}-\frac{\mathrm{I} \sqrt{3}}{2}\right)^{\frac{n}{2}}}{4}\right) \chi_{\left\{\mathit{modp} \left(n ,2\right)=0\right\}}-\frac{9 \chi_{\left\{\mathit{modp} \left(n ,3\right)=0\right\}}}{32}+\frac{\mathrm{I} \sqrt{3}\, \left(\mathrm{-I}\right)^{\frac{n}{3}-\frac{2}{3}} \chi_{\left\{\mathit{modp} \left(n ,3\right)=2\right\}}}{4}+\frac{\left(-\frac{1}{2}-\frac{\mathrm{I} \sqrt{3}}{2}\right)^{\frac{n}{4}} \chi_{\left\{\mathit{modp} \left(n ,4\right)=0\right\}}}{2}+\left(-\frac{27}{32}-\frac{\mathrm{I} \sqrt{3}}{8}\right) \left(-1\right)^{\frac{n}{6}} \chi_{\left\{\mathit{modp} \left(n ,6\right)=0\right\}}-\frac{\sqrt{3}\, \left(-1\right)^{\frac{n}{6}-\frac{5}{6}} \chi_{\left\{\mathit{modp} \left(n ,6\right)=5\right\}}}{4}
\end{dmath}
\end{enumerate}
\QEDA
\end{example}

\begin{example}[Collatz sequence beginning at 21, OEIS:\href{https://oeis.org/A033481}{A033481}]\label{ex:oeis2} The general term $s_n$ of the sequence is defined by the recursion
\begin{equation}
    s_0=21, \, s_{n+1}=\begin{cases}\frac{s_n}{2}\,\quad\,\phantom{1111} \text{ if } s_n\equiv 0 \mmod 2\\
                          3\,s_{n} + 1\, \quad \text{ if } s_n\equiv 1 \mmod 2
            \end{cases}\, \, \text{ for all } n\geq 1.
\end{equation}
The generating function of $(s)_n$ is given by
\[f(z)\coloneqq \frac{-7 z^{7}-14 z^{6}-28 z^{5}-56 z^{4}-5 z^{3}+32 z^{2}+64 z +21}{\left(1-z \right) \left(z^{2}+z +1\right)}.\]
This function is a non-proper hypergeometric-type function as \texttt{FPS} \cite{FPS} finds the power series formula
\[f(z)\coloneqq 7 z^{4}+14 z^{3}+28 z^{2}+63 z +19+\moverset{\infty}{\munderset{n =0}{\textcolor{gray}{\sum}}}\! 4 z^{n}+\moverset{\infty}{\munderset{n =0}{\textcolor{gray}{\sum}}}\! -2 z^{3 n} + \moverset{\infty}{\munderset{n =0}{\textcolor{gray}{\sum}}}\! -3 z^{3 n +1}.\]
Thus $(s)_n\notin (\mathcal{H}_T)$. However, if we remove the polynomial part from the expansion, i.e., we consider 
\[g(z)\coloneqq f(z)-\left(7 z^{4}+14 z^{3}+28 z^{2}+63 z +19\right),\]
then the resulting sequence of coefficients is of hypergeometric type. The formula can be deduced either with \texttt{FPS} or its `child' \texttt{HTS}. As we removed the polynomial part in the expansion of $f(z)$, the new sequence is $(u)_n = (s)_n-(19,63,28,14,7,0,0,\ldots)$. Using \texttt{FPS:-FindRE} we find the following recurrence equation satisfied by $u_n$:
\begin{lstlisting}
> RE:=FPS:-FindRE(f-(7*z^4 + 14*z^3 + 28*z^2 + 63*z + 19),z,u(n))
\end{lstlisting}
\begin{dmath}\label{eq:oeis2eq}
\mathit{RE} \coloneqq \left(-n +1\right) u \! \left(n \right)+\left(4 n -12\right) u \! \left(n -4\right)+\left(n -1\right) u \! \left(n -3\right)+\left(2 n +2\right) u \! \left(n -2\right)+\left(-4 n +12\right) u \! \left(n -1\right)+\left(-2 n -2\right) u \! \left(n +1\right)=0.
\end{dmath}
Hence the formula
\begin{lstlisting}
> REtoHTS(RE,u(n),[2, 1, 4, 2, 1, 4])
\end{lstlisting}
\begin{dmath}\label{eq:oies2res}
4-2 \chi_{\left\{\mathit{modp} \left(n ,3\right)=0\right\}}-3 \chi_{\left\{\mathit{modp} \left(n ,3\right)=1\right\}}.
\end{dmath}
The main point in this example is that formulas of solutions to holonomic recurrence equations can be found with enough initial values. We usually prefer to supply a procedure instead of a list of values, as this will allow the code to use as many values as necessary. For this example, the syntax would be:
\begin{lstlisting}
> U:=proc(n) U(n):=subs([n=n-1,u=U],solve(RE,u(n+1))) end proc: 
U(0):=2:U(1):=1:U(2):=4:U(3):=2:U(4):=1:U(5):=4:
> REtoHTS(RE,u(n),U):
\end{lstlisting}
We hid the output as it is precisely \eqref{eq:oies2res}.
 
\QEDA
\end{example}

\section{Conclusion}
In conclusion, this article introduced hypergeometric-type sequences with a formalism of interlacement described by $m$-fold indicator sequences. We showed that these sequences are generated by proper hypergeometric-type series. It may be possible to generalize the study to include proper Laurent-Puiseux series of hypergeometric type. For Puiseux series, the corresponding interlacements may be viewed as $\alpha$-fold indicator sequences for some $\alpha\in\QQ$.

We proved that $\mathcal{H}_T$ is a ring and presented an algorithm to decide whether a given holonomic term is of hypergeometric type or not. The latter comes as a complement of the algorithms by Petkov\v{s}ek and van Hoeij \cite{petkovvsek1992hypergeometric,van1999finite} to detect when a given holonomic term can be written as a linear combination of interlaced hypergeometric terms.

It is worth mentioning that C-finite sequences \cite{zeilberger1990holonomic}, also called LRS (linear recurrence sequence) \cite{ouaknine2012decision}, form a subclass of hypergeometric-type sequences. The inclusion is immediate from their writing as exponential polynomials \cite{chonev2023zeros}. Ouaknine and Worrell showed that one can decide if any C-finite sequence of order $5$ or less is positive \cite{ouaknine2014positivity}. Could the same conclusion hold for hypergeometric-type sequences that satisfy holonomic recurrence equations of order at most $5$? The target is, of course, a particular case (see \Cref{prop:holonomicity}) of the general class of holonomic sequences for which the positivity problem is only partially studied \cite{kauers2010can,pillwein2013termination,ibrahim2023positivity}.

We end with an observation concerning the generating functions of $m$-fold indicator sequences. It is easy to see that 
\begin{equation}\label{eq:genmfoldind}
    f_{m,j}(z)\coloneqq \frac{z^j}{1-z^{m}} = \sum_{n=0}^{\infty} \mfoldind{j}{m} z^n,\, \quad\, m,j\in\NN,\, j<m.
\end{equation}
As $m$-fold indicator sequences may be regarded as a basis of a free module, it sounds interesting to study the structure of proper hypergeometric-type functions and see their relation to the $f_{m,j}$'s.

\medskip

\textbf{Acknowledgment.} The author thanks Wolfram Koepf for encouraging this work. He is grateful to Jo\"{e}l Ouaknine for the good time spent at the Max Planck Institute for Software Systems, where he started this article. The author was supported by UKRI Frontier Research Grant EP/X033813/1.

\end{document}